\documentclass[runningheads]{llncs}

\usepackage{graphics} 
\usepackage{epsfig} 
\usepackage{mathptmx} 
\usepackage{times} 
\usepackage{amsmath} 
\usepackage{amssymb}  
\usepackage{cuted}
\usepackage{float}
\usepackage{color}
\usepackage{url}
\usepackage{bm}

\usepackage[ruled,vlined,linesnumbered,noresetcount]{algorithm2e}
\usepackage{tabularx}

\begin{document}
\title{
Farsighted Risk Mitigation of Lateral Movement Using Dynamic Cognitive Honeypots
}
\titlerunning{Cognitive Honeypots against Persistent Lateral Movement}
%

\author{Linan Huang \and
Quanyan Zhu\thanks{This research is partially supported by awards ECCS-1847056, CNS-1544782, CNS-2027884, and SES-1541164 from National Science of Foundation (NSF), and grant W911NF-19-1-0041 from Army Research Office (ARO).} }
\authorrunning{L. Huang and Q. Zhu}
%
\institute{Department of Electrical and Computer Engineering, New York University\\
 2 MetroTech Center, Brooklyn, NY, 11201, USA \\
\email{\{lh2328,qz494\}@nyu.edu}
}
\maketitle              
\begin{abstract}
Lateral movement of advanced persistent threats has posed a severe security challenge. 
Due to the stealthy and persistent nature of the lateral movement, defenders need to consider time and spatial locations holistically to discover latent attack paths across a large time-scale and achieve long-term security for the target assets. In this work, we propose a time-expanded random network to model the stochastic service links in the user-host enterprise network and the adversarial lateral movement. We design cognitive honeypots at idle production nodes and disguise honey links as service links to detect and deter the adversarial lateral movement. The location of the honeypot changes randomly at different times and increases the honeypots' stealthiness. Since the defender does not know whether, when, and where the initial intrusion and the lateral movement occur, the honeypot policy aims to reduce the target assets' Long-Term Vulnerability (LTV) for proactive and persistent protection. 
We further characterize three tradeoffs, i.e., the probability of interference, the stealthiness level, and the roaming cost. 
To counter the curse of multiple attack paths, we propose an iterative algorithm and approximate the LTV with the union bound for computationally efficient deployment of cognitive honeypots. 
The results of the vulnerability analysis illustrate the bounds, trends, and a residue of LTV when the adversarial lateral movement has infinite duration. 
Besides honeypot policies, we obtain a critical threshold of compromisability to guide the design and modification of the current system parameters for a higher level of long-term security. We show that the target node can achieve zero vulnerability under infinite stages of lateral movement if the probability of movement deterrence is not less than the threshold.



\end{abstract}

\section{Introduction}
Advanced Persistent Threats (APTs) have recently emerged as a critical security challenge to enterprise networks.  Their stealthy, persistent, and sophisticated nature has made it difficult to prevent, detect, and deter them.
The life cycle of APT attacks consists of multiple stages and phases \cite{ATTACK,zhu2018multi}.  After the initial intrusion by phishing emails, social engineering, or an infected USB, an attacker can enter the enterprise network from an external network domain. 
Then, the attacker establishes a foothold, escalates privileges, and moves laterally in the enterprise network to search for valuable assets as his final target. The targeted assets can be either a database with confidential information or a controller in an industrial plant as shown in the instance of APT27 \cite{APT27} and Stuxnet, respectively. 
Valuable assets are usually segregated and cannot be compromised by an attacker directly from the external domain in the initial intrusion phase. Therefore, it is indispensable for the attacker to exploit the internal network flows of legitimate service links between hosts and users to move laterally from the location of the initial intrusion to the final target of valuable assets.

Early detection of the adversarial lateral movement is challenging. 
First, an APT attacker is persistent. 
The long duration between the initial intrusion and the final target compromise makes it difficult for the defender to relate alarms over a time scale of years and piece together shreds of evidence to identify the attack path. 
Second, an APT attack is stealthy. 
Each time the attacker has compromised a new network entity, such as a host, and obtained its root privilege, he does not take any subversive actions on the compromised entity and remains ``under the radar''. 
These entities are only used as the attacker's stepping stones toward the final target. 
Third, the high volume of network traffic during regular operation generates a considerable number of false alarms, and thus significantly delays and reduces the accuracy of adversary detection. 
Without an accurate and timely detection of adversarial lateral movement, defensive methods, such as patching and frequent resetting  of  suspicious entities, become cost-prohibitive and significantly reduce  operational efficiency as those entities become unavailable for the incoming service links.

Honeypot is a promising active defense method of deception. A honeypot is a  monitored and regulated trap that is disguised to be a valuable asset for the attacker to compromise.
 Since legitimate users do not have the motivation to access a honeypot, any inbound network traffic directly reveals the attack with negligible false alarms. 
The off-the-shelf honeypots are applied at fixed locations and on isolated machines that are not involved in the regular operation.  
Honeypots at fixed locations are easy to implement. Isolating the honeypot completely from the production system can reduce the risk that an attacker uses the honeypot as a pivot node to penetrate the production system \cite{spitzner2003honeypots}. 
Despite the advantages, honeypots at fixed and isolated locations can be easily identified by sophisticated attackers \cite{krawetz2004anti} and become ineffective. 
Motivated by the concept of cognitive radio \cite{mitola1999cognitive} and roaming honeypots \cite{khattab2004roaming}, we develop the concept of cognitive honeypots to mitigate the Long-Term Vulnerability (LTV) of a target asset during the adversarial lateral movement. 
Contrary to the off-the-shelf honeypots, the cognitive honeypots aim to leverage idle machines of the production system and configure them into honeypots to make the deception indecipherable and unpredictable for the attacker. 
Since the defender reconfigures part of the production systems into honeypots, she needs to guarantee that the honeypot configuration  does not interfere with service links. 
Also, the defender needs to balance the utility of security with the cost of reconfiguration. 
We manage to consider the above three factors, i.e., the level of stealthiness/indecipherability, the probability of interference, and the cost of roaming, in determining the optimal honeypot policy that minimizes the target asset's LTV. 

In this work, we model the adversarial lateral movement in the enterprise network as a time-expanded network \cite{casteigts2012time}, where the additional temporal links connect the isolated spatial service links across a long time to reveal persistent attack paths explicitly. 
We consider the scenario where service links occur randomly at each stage and the attacker can exploit these service links for lateral movement with a success probability. 
Due to the \textit{curse of multiple attack paths}, the computation complexity increases dramatically with the network size and the number of stages. 
To efficiently compute the optimal policy for the cognitive honeypot, we propose an iterative algorithm and approximate the LTV by its upper and lower bounds, which result in the optimal conservative and risky honeypot policies, respectively. 
The results of the  vulnerability analysis illustrate the limit and the bounds of LTV when the duration of lateral movement goes to infinity under direct and indirect policies, respectively. 
Without proper mitigation strategies, vulnerability never decreases over stages and the target node is doom to be compromised. Under the improved honeypot strategies, a \textit{vulnerability residue} exists and LTV cannot be reduced to $0$. 
Besides honeypot policies, we further investigate the possibility of changing the frequency of service links and the probability of successful compromise for \textit{long-term security}. We manage to character a critical threshold for the \textit{Probability of Movement Deterrence (\textbf{PoMD})} and prove that the target node can achieve zero vulnerability even when the adversarial lateral movement last for infinite stages if POMD is not less than the threshold.


\subsection{Related Works}

\subsubsection{Lateral Movement Detection and Mitigation}
Various methods have been proposed for lateral movement detection \cite{liu2018latte,tian2019real,lah2018proposed}. 
However, most of them rely on accurate and timely identification of the initial intrusion, which may be challenging to achieve. 
Mitigation methods of network topology change have also been proposed to delay  lateral movement \cite{noureddine2016game} and reduce its adversarial impact \cite{purvine2016graph}. Authors in \cite{HuangAPT,huang2018PER} have proposed a proactive defense-in-depth model against the multi-stage multi-phase attacks.  
Previous works have also analyzed security metrics, such as reachability \cite{purvine2016graph}, 
enforceability \cite{alsaleh2018verifying}, and survivability \cite{shi2019quantitative}, to reduce risk and loss under lateral movement attacks. 
Compared to these works, our work applies honeypots and honey links to detect and mitigate lateral movement. Moreover, we enable the analysis of the target's LTV under an undetected initial intrusion and an arbitrary duration of lateral movement.

\subsubsection{Cognitive Honeypots}
Honeypots as a defensive deception method have been widely studied in the literature.
The authors in \cite{pawlick2019optimal,huang2019adaptive,huang2018analysis} have investigated the optimal timing and actions to attract and engage attackers in the honeypot. The authors in \cite{huang2020game} have investigated the optimal honeypot configuration and the signaling mechanism to simultaneously incentivize attackers and disincentivize legitimate users to access a honeypot. 
All these honeypots are assumed to be placed at fixed and segregated locations. 
In this work, we consider cognitive honeypots that use the idle machines of the production system to increase the stealthiness of honeypots. 
The terminology of ``cognitive honeypots'' has appeared in \cite{goldberg2017cognitive} but refers to a cognition of the suspicion level.
The authors in \cite{horak2019optimizing} have investigated the optimal honeypot locations during the adversarial lateral movement to prevent the attacker from compromising the target node. 
Their honeypot policy requires a partial observation of the state, which may not be available as a result of the attacker's stealthiness. 
Our work assumes that the defender does not know whether, when, or where the initial intrusion and the lateral movement occur in the network. 
Without real-time feedback information such as alerts of node compromise, the cognitive honeypot provides proactive and persistent protection of the valuable asset.

\subsubsection{Time-Expanded Network}
Time-expanded networks have been applied in transportation \cite{wang2015efficient}, 
 satellite communications \cite{jiang2020reinforcement},  and network security \cite{xu2019cybersecurity}. 
Since the transportation planning and satellite communications follow a timetable, the time-expanded networks in these applications usually have time-varying links that are deterministic and known at all stages. 
In enterprise networks, the defender does not know which service links will be used in the ensuing stages. 
Thus, we consider a time-expanded network with random topology. 
Compared to attack graphs (e.g., \cite{kaynar2016taxonomy}), which focus on capturing the paths of an attack, 
the time-expanded network explicitly portrays the timing of the attacks and captures the temporal information of the legitimate network flows and the adversarial lateral movement.  


\subsection{Notation and Organization of the Paper}
Throughout this paper, we use the pronoun `he' for the attacker and `she' for the defender. The superscript represents the time index.  The calligraphic letter $\mathcal{V}$ represents a set and $\mathcal{V} \setminus \mathcal{V}_I$ means the set of elements in $\mathcal{V}$ but not in $\mathcal{V}_I$.  
We summarize important notations in Table \ref{table:notation} for readers' convenience. 
\begin{table}[]
\centering
\caption{Summary of notations. 
\label{table:notation}}
\begin{tabular}{ll}
\hline
\textbf{Variable}      & \textbf{Meaning}      \\ \hline
$\mathcal{V}=\{\mathcal{V}_U,\mathcal{V}_H\}$ & Node set of users and hosts \\
$N=|\mathcal{V}|$ & Number of user and host nodes \\
$\mathcal{V}_I \subseteq \mathcal{V}$ & Demilitarized Zone (DMZ), i.e., the node-set of potential initial intrusion  \\
$\mathcal{V}_D\subseteq \mathcal{V}$ & The node-set that can be reconfigured as honeypots  \\
$\mathcal{V}_S$ & The set of all the subsets of $\mathcal{V}$ \\
$n_{j_0}\in \mathcal{V}\setminus\mathcal{V}_I$ & The target node that contains valuable assets \\ 
$\Delta k\in \mathbb{Z}_0^+$ & The length of the adversarial lateral movement\\
$\rho_i$ & The probability that the initial intrusion occurs at node $n_i\in \mathcal{V}_I$ \\ 
$\beta$ & The probability / frequency of service links \\
$\lambda$ & The probability of a successful compromise \\
$\gamma$ & The probability of honey links \\
$q_{i,j}$ & The probability that the attacker identifies the honey link from $n_i$ to $n_j$ \\
\hline
\end{tabular}
\end{table}
The rest of the paper is organized as follows. Section \ref{sec:model} introduces the time-expanded network to model the random arrival of the service links, the adversarial lateral movement, and the implementation of cognitive honeypots. 
 In Section \ref{sec:riskminimization}, we compute the optimal honeypot policy dependent on the level of stealthiness, the probability of interference, and the cost of roaming. The LTV of the target node is then analyzed. Section \ref{sec:conclusion} concludes the paper.

\section{Chronological Enterprise Network Model}
\label{sec:model}
 \begin{figure*}[t]
\centering
\includegraphics[width=0.80 \textwidth]{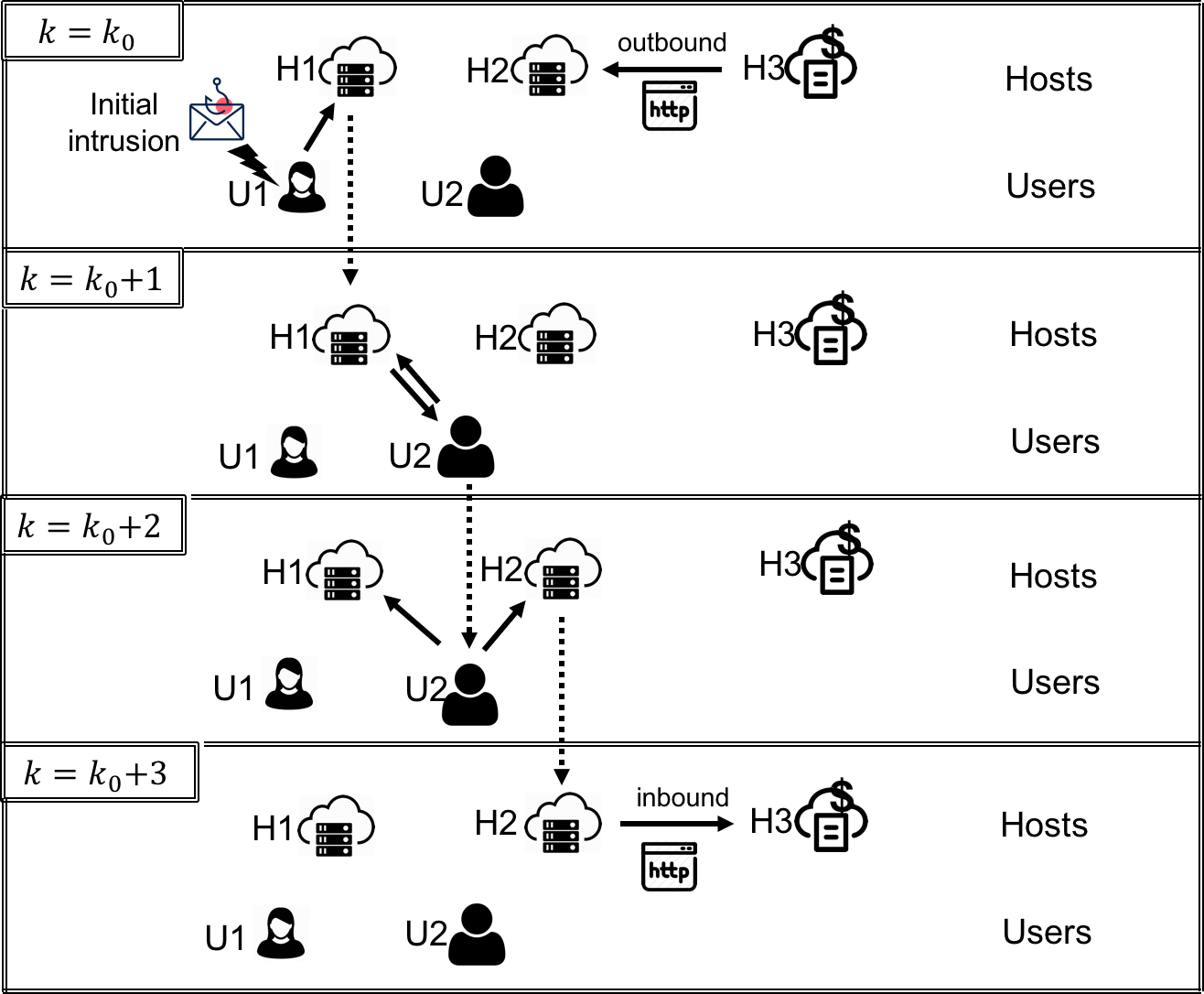}
\caption{ 
A sequence of user-host networks with service links in chronological order under  discrete stage-index $k$. 
The initial stage $k_0$ is the stage of the attacker's initial intrusion yet the defender does not know the value of $k_0$. 
The solid arrows show the direction of the user-host and host-host network flows. By incorporating part of temporal links denoted by the dashed arrows, we reveal the attack path over a long period explicitly. 
}
\label{fig:ScenarioDiag}
\end{figure*}  
We model the normal operation of an enterprise network over a continuous period as a sequence of user-host networks in chronological order. 
As shown in Fig. \ref{fig:ScenarioDiag}, nodes U$1$ and U$2$ represent the two users' client computers. 
Nodes H$1$, H$2$, and H$3$ represent three hosts in the network. 
In particular, host H$3$ stores confidential information or controls a critical actuator, thus the defender needs to protect H$3$ from attacks. 
Define $\mathcal{V}:=\{\mathcal{V}_U, \mathcal{V}_H\}$ as the node set where $\mathcal{V}_U, \mathcal{V}_H$ are the sets of the user nodes and hosts, respectively. 
The solid arrows represent two types of service links, i.e., the user-host connections and the host-host communications through an application such as HTTP \cite{chen2019enterprise}. 
Users such as U$1$ and U$2$ can access non-confidential hosts, such as H$1$ and H$2$, through their client computers for upload and/or download. 
However, to prevent data theft and physical damages, host H$3$ is inaccessible to users; e.g., there are no service links from U$1$ or U$2$ to H$3$ at any stage $k$. 
Since the normal operation requires data exchanges among hosts,  directed network flows exist among hosts at different stages; e.g., H$3$ has an outbound connection to H$2$ at stage $k=k_0$ and an inbound connection from H$2$ at stage $k=k_0+3$. 
We assume that both types of service links occur randomly and last for a random but finite duration. Whenever there is a change of network topology, i.e., adding or deleting the user-host and host-host links, we define it as a new stage. 
We can characterize the chronological network as a series of user-host networks at discrete stages $k=k_0,k_0+1,\cdots,k_0+\Delta k$, where the initial stage $k_0\in \mathbb{Z}^+$ and $\Delta k\in \mathbb{Z}_0^+$. 
Since APTs are stealthy, the defender may not know the value of $k_0$, i.e., when the initial intrusion happens or has already happened. The lack of accurate and timely identification of the initial intrusion brings a significant challenge to detect and deter the lateral movement.

\subsection{Time-Expanded Network and Random Service Links}
\label{subsec:Time-Expanded Network}
We abstract the discrete series of networks in Fig. \ref{fig:ScenarioDiag} from $k\in \{k_0,\cdots,k_0+\Delta k\}$ as a time-expanded network $\mathcal{G}=(\mathcal{V},\mathcal{E},\Delta k)$ in Fig. \ref{fig:TVG}. 
\begin{figure*}[t]
\centering
\includegraphics[width=0.80 \textwidth]{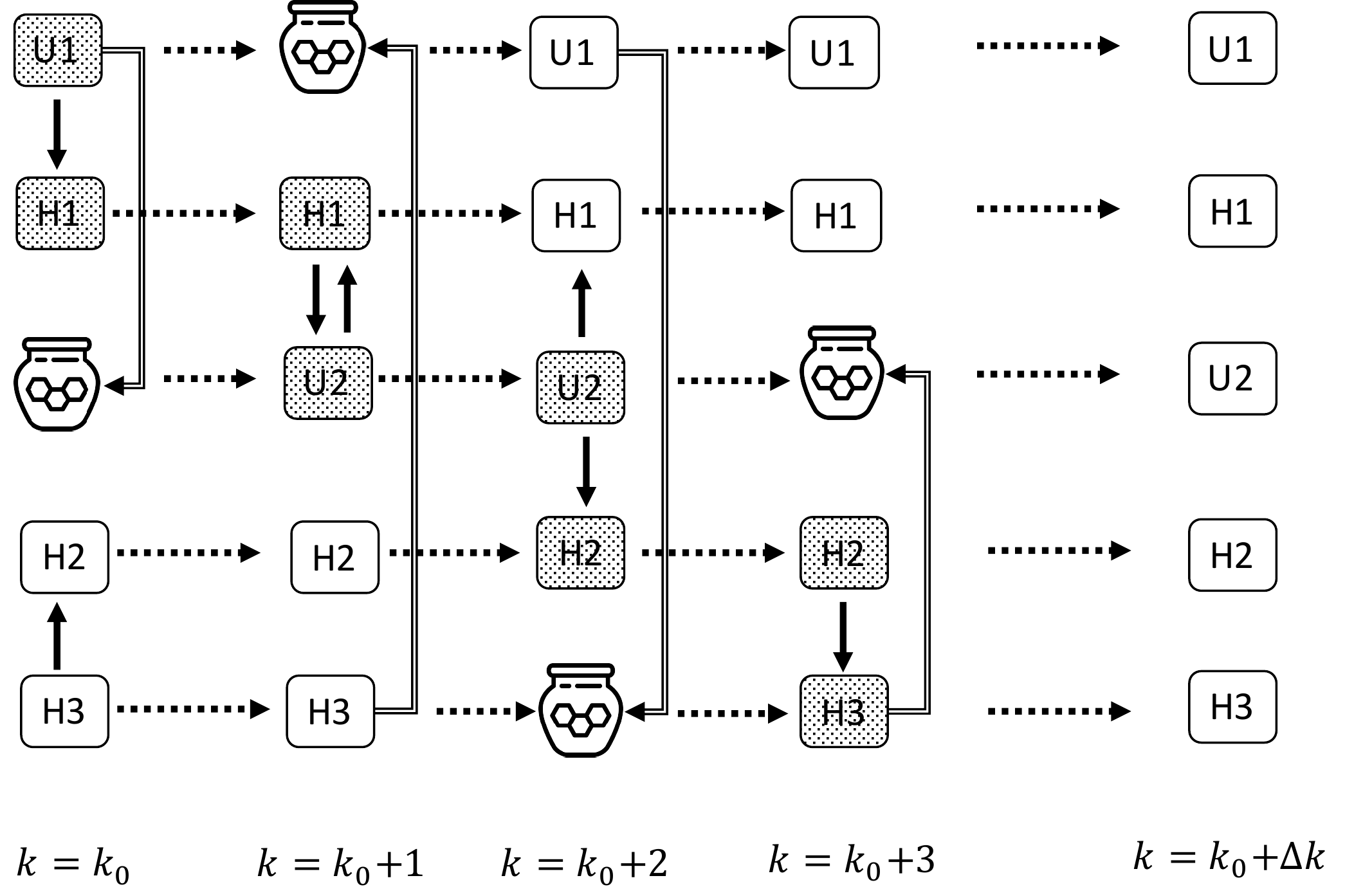}
\caption{ 
Time-expanded network $\mathcal{G}=\{\mathcal{V},\mathcal{E},\Delta k\}$ for the adversarial lateral movement and the cognitive honeypot configuration. 
The solid, dashed,  double-lined  arrows represent the service links, the temporal connections, and the honey links to honeypots, respectively. The shadowed nodes reveal the attack path from U$1$ to H$3$ explicitly over $\Delta k=3$ stages. 
}
\label{fig:TVG}
\end{figure*}  
In the time-expanded network, we distinguish the same user or host node by the stage $k$ and define $n_i^k\in \mathcal{V}$ as the $i$-th node in set $\mathcal{V}$ at stage $k\in \{k_0,\cdots,k_0+\Delta k\}$. 
We drop the superscript $k$ if we refer to the node rather than the node at stage $k$ or the time does not matter. 
We can assume without loss of generality that the number of nodes  $N:=|\mathcal{V}|$ does not change with time as we can let $\mathcal{V}$ contain all the potential users and hosts in the enterprise network over $\Delta k$ stages. 
The link set $\mathcal{E}:=\{\mathcal{E}^{k_0},\cdots,\mathcal{E}^{k_0+\Delta k} \} \cup \{\mathcal{E}_C^{k_0},\cdots, \mathcal{E}_C^{k_0+\Delta k-1} \}$ consists of two parts. 
On the one hand, the user-host and  host-host connections at each stage $k\in \{k_0,\cdots,k_0+\Delta k\}$ are represented by the set $\mathcal{E}^k=\{e(n_i^k,n_j^k)\in \{0,1\} | n_i^k,n_j^k\in \mathcal{V}, i\neq j, \forall i,j\in \{1,\cdots, N\} \} $. 
On the other hand, set $\mathcal{E}_C^k:=\{e(n_i^k,n_i^{k+1})=1 | n_i^k,n_i^{k+1}\in \mathcal{V}, \forall i\in \{1,\cdots, N\} \}$ contains the virtual temporal links from stage $k$ to $k+1$. 
A link exists if $e(\cdot,\cdot)=1$ and does not if  $e(\cdot,\cdot)=0$.  
The time-expanded network $\mathcal{G}$ is   a directed graph due to the temporal causality represented by the set $\mathcal{E}_C^k, k\in \{k_0,\cdots,k_0+\Delta k-1\}$. 

Since the user-host and the host-host connections happen randomly at each stage, we assume that a service link from node $n_i^k \in \mathcal{V}$ to node $n_j^k\in \mathcal{V}\setminus \{n_i^k\}$ exists with probability $\beta_{i,j}\in [0,1]$ for any stage $k\in \{k_0,\cdots,k_0+\Delta k\}$. 
If a connection from node $n_i^k$ to $n_j^k$ is prohibitive; e.g., U$1$ cannot access H$3$ in Fig. \ref{fig:ScenarioDiag}, then $\beta_{i,j}=0$. 
We can define $\beta:=\{\beta_{i,j}\},i,j\in\{1, \cdots, N\}$, as the service-link generating matrix without loss of generality by letting $\beta_{i,i}=0, \forall i\in \{1, \cdots, N\}$. In this work, we consider a time-invariant $\beta$ whose value can be estimated empirically from long-term historical data\footnote{For example, we can use the user-computer authentication dataset from the Los Alamos National Laboratory enterprise network \cite{hagberg-2014-credential} to estimate the probability of user-host service links over a long period. The dataset is available at \url{https://csr.lanl.gov/data/auth/}.}.
The service links at each stage may only involve a small number of nodes and leave other nodes idle.
\begin{definition}
\label{def:idle}
A node $n_i^k\in \mathcal{V}$ is said to be \textbf{idle} at stage $k$ if it is neither the source nor the sink node of any service link 
at stage $k$, i.e., $e(n_i^k,n_j^k)=0,e(n_j^k,n_i^k)=0, \forall n_j^k\in \mathcal{V}$. 
\end{definition}

\subsection{Attack Model of Lateral Movement over a Long Duration}
\label{subsec:attack model}
We assume that the initial intrusion can only happen at a subset of $N$ nodes $\mathcal{V}_I \subseteq \mathcal{V}$ due to the network segregation. We can refer to $\mathcal{V}_I$ as the Demilitarized Zone (DMZ). 
Take Fig. \ref{fig:ScenarioDiag} as an example, if all hosts in the enterprise network are segregated from the Internet, the initial intrusion can only happen to the client computer of U$1$ or U$2$ through phishing emails or social engineering. 
Although network segregation narrows down the potential location of initial intrusion from $\mathcal{V}$ to the subset $\mathcal{V}_I$ that may contain only one node, it is still challenging for the defender to prevent the nodes in $\mathcal{V}_I$ from an initial intrusion as the defender cannot determine \textit{when} the initial intrusion happens; i.e., the value of $k_0$ is unknown. 
In this work, we assume that the initial intrusion only happens to one node in set $\mathcal{V}_I$ at a time; i.e., no concurrent intrusions happen. 
Once the attacker has entered the enterprise network via the initial intrusion from an external network domain, he does not launch new intrusions from the  external domain to compromise more nodes in $\mathcal{V}_I$. Instead, the attacker can exploit the internal service links to move laterally over time, which is much  stealthier than intrusions from external network domains. 
 For example, after the attacker has controlled U$1$'s computer by phishing emails, he would not send phishing emails to other users from the external network domain, which increases his probability of being detected. 
We define $\rho_i\in [0,1]$ as the probability that the initial intrusion happens at node $n_i^{k_0} \in \mathcal{V}_I, \forall k_0\in \mathbb{Z}^+$. The probability satisfies $\sum_{i\in \mathcal{V}_I } \rho_i=1$ and is assumed to be independent of the stage $k_0$. 
This probability of initial intrusion can be estimated based on the node's vulnerability assessed by historical data, red team exercises, and the Common Vulnerability Scoring System (CVSS) \cite{mell2006common}.

After the initial intrusion, the attacker can exploit service links at different  stages by various techniques to move laterally, such as Pass the Hash (PtH), taint shared content, and remote service session hijacking \cite{ATTACK}. 
Take PtH as an example, when a user enters the password and logs into host H$1$ from a compromised client computer U$1$ at stage $k_0$ as shown in Fig. \ref{fig:ScenarioDiag}, the attacker at U$1$ can capture the valid password hashes for accessing host H$1$ by credential access technique. Then, the attacker can use the captured hashes to access the host H$1$ for all the future stage $k>k_0$. 
The attacker can also compromise a user node from a compromised host by tainting the shared content, i.e., adding malicious scripts to valid files in the host. Then, the malicious code can be executed when user U$2$ downloads those files from H$1$ at stage $k_0+1$. 
PtH (resp. tainting shared content) enables an adversarial lateral movement from a user node (resp. host node) to a host node (resp. user node). The attacker can also use remote service session hijacking, such as Secure Shell (SSH) hijacking and Remote Desktop Protocol (RDP) hijacking, to move laterally between hosts by hijacking the inbound or outbound network flows. 
In this work, we assume that once the attacker compromises a node, he retains the control of the node for the given length of time window $\Delta k$ determined by the defender. For example, the defender can require users to update their password every $\Delta k$ days to invalidate the PtH attack. 
During the time window, i.e., from the initial intrusion $k=k_0$ to $k=k_0+\Delta k$, the attacker can launch simultaneous attacks from all the compromised nodes to move laterally whenever there are outbound service links from them.  
If there are multiple service links from one compromised node, the attacker can also compromise all the sink nodes of these service links within the stage.  
Note that the only objective of the attacker is to search for valuable nodes (e.g., H$3$), compromise it, and then launch subversive attacks for data theft and physical damages. 
Thus, we assume that the attack does not launch any subversive attacks in all the compromised nodes except at the target node to remain stealthy.  
That is, even though the attacker retains the control of the compromised nodes, he only uses them as stepping stones to reach the target node.

The persistent lateral movement over a long time period enables the attacker to reach and compromise segregated nodes that are not in the DMZ $\mathcal{V}_I$. 
In both Fig. \ref{fig:ScenarioDiag} and Fig. \ref{fig:TVG}, although the network has no direct service links, represented by solid arrows, from U$1$ to H$3$ at each stage, the cascade of \textit{static security} in all stages does not result in \textit{long-term security} over $\Delta k=3$ stages.  
After we add the temporal links represented by the dashed arrows and consider stages and spatial locations holistically, we can see the attack path from the initial intrusion node U$1$ to the target node H$3$ over $\Delta k=3$ stages as highlighted by the shadows in Fig. \ref{fig:TVG}. 
The temporal order of the service links affects the likelihood that the attacker can compromise the target node. 
For example, if we exchange the services links that happen at stage $k_0+1$ and  stage $k_0+2$, then the attacker from node U1 cannot reach H$3$ in $\Delta k=3$ stages. Since the attacker can launch simultaneous attacks from multiple compromised nodes to move laterally, there can exist multiple attack paths from an initial intrusion node to the target node. 

The adversarial exploitation of service links is not always successful due to the defender's mitigation technologies against lateral movement techniques \cite{ATTACK}. For example, the firewall rules to block  RDP traffic between hosts can invalidate RDP hijacking. 
If the attacker has compromised nodes $n^{k'}_i\in \mathcal{V}$ before stage $k>k'$ and a service link from $n_i^k$ to $n_j^k\in \mathcal{V} \setminus \{n_i^k\}$ exists at stage $k$, i.e., $e(n_i^k,n_j^k)=1$, we can define $\lambda_{i,j}\in [0,1]$ as the probability that the attacker at node $n_i^k$ successfully compromises node $n_j^k$, which is assumed to be independent of stage $k$. 

\subsection{Cognitive Honeypot} 
\label{subsec:cog honeypot}
The lateral movement of persistent and stealthy attacks makes the enterprise network insecure in the long run. 
The high rates of false alarms and the miss detection of both the initial external intrusion and the following internal compromise make it challenging for the defender to identify the set of nodes that have been compromised. 
Thus, the defender needs to patch and reset all suspicious nodes at all stages to deter the attacks, which can be cost-prohibitive. 

Honeypots are a promising active defense method to detect and deter these persistent and stealthy attacks by deception \cite{nawrocki2016survey}. In this paper, the connection from a service node to a honeypot is referred to as a honey link. The defender disguises a honey link as a service link to attract attackers.  For example, the defender can start a session with remote services from a host to a honeypot. The attacker who has compromised the host will be detected once he hijacks the remote service session and carries out actions in the honeypots. 
Since regular honeypots are implemented at fixed locations and on machines that are never involved in the regular operation, advanced attacks like APTs can identify the honeypots and avoid accessing them. 
Motivated by the roaming honeypot \cite{khattab2004roaming} and the fact that the service links at each stage only involve a small number of nodes, we develop the following cognitive honeypot configuration that utilizes and reconfigures different idle nodes at different stages as honeypots. 
Let $\mathcal{V}_D\subseteq \mathcal{V}$ be the subset of nodes that can be reconfigured as honeypots when idle. 
At each stage $k$, the defender randomly selects a node $n_w^k\in \mathcal{V}_D$ to be the potential honeypot and creates a random honey link from other nodes to $n_w^k$. 
Since disguising a honeypot as a normal node requires emulating massive services and the continuous monitoring of all inbound network flows are costly, 
we assume that the defender sets up at most one honeypot and monitors one honey link at each stage. 

As shown in Fig. \ref{fig:TVG}, U$1$, H$2$, and H$3$ are idle at stage $k_0+1$ and U$1$ is reconfigured as the honeypot. The link from H$3$ to U$1$ is the honey link which is monitored by the defender. 
At stage $k_0$, U$2$ is the only idle node and is reconfigured as the honeypot with a honey link from U$1$ to U$2$. 
As stated in Section \ref{subsec:attack model}, the attacker who has compromised U$1$ at stage $k_0$ remains stealthy and does not sabotage any normal operations.  Thus, the defender can reconfigure U$1$ as a honeypot at stage $k_0+1$. However, the honeypot of U1 at stage $k_0+1$ cannot identify the attacker by monitoring all the inbound traffic as he has already compromised U$1$. 
On the contrary, the honeypots at stage $k_0$ and $k_0+2$ can trap the attackers who have compromised U$1$ and mistaken the honey links as service links\footnote{The defender would avoid configuring honey links from the target node to the honeypot. If the attacker has not compromised the target node H$3$ as shown in stage $k_0+1$, the honeypot cannot capture the attacker. If the attacker has compromised the target node as shown in stage $k_0+3$, then the late detection cannot reduce the loss that has already been made.}. 
Theoretically, the honeypot can achieve zero false alarms as the legitimate network flows should occur only at the service links. For example, although the existence of the honey link at stage $k_0$ enables legitimate users at U$1$ to access another user's computer U$2$, a legitimate user aiming to finish the service link from U$1$ to H$1$ should not access any irrelevant nodes other than host H$1$. On the other hand, an attacker at U$1$ cannot tell whether the links from U$1$ to H$1$ and U$2$ are service links or honey links. Thus, only an attacker at U$1$ can access the honeypot U$2$ at stage $k_0$. 

\subsubsection{Random Honeypot Configuration and Detection}
Since the defender can neither predict future service links nor determine the set of compromised nodes at the current stage, she needs to develop a time-independent policy $\gamma:=\{\gamma_{l,w} \}, \forall n^k_l,n^k_w\in \mathcal{V}$, to determine the honeypot location and the honey link at each stage $k$ to minimize the risk that an attacker from the node of the initial intrusion can compromise the target node after $\Delta k$ stages. 
Each policy element $\gamma_{l,w}$ is the probability that the honeypot is node $n_w^k$ and the honey link is from node $n_l^k$ to $n_w^k$ at stage $k\in \{k_0, \cdots, k_0+ \Delta k\}$. 
Note that $\gamma_{i,i}=0, \forall i\in \mathcal{V}$, and 
we can let  $n_l,n_w$ belong to the entire node set $\mathcal{V}$ without loss of generality because   
if a node $n_w \notin \mathcal{V}_D$ is not reconfigurable,  then we can let the probability $\gamma_{l,w}$ be zero. 
Define $n_{j_0}\in \mathcal{V}\setminus \mathcal{V}_I$ as the target node to protect for all stages and the target node is segregated from the set of potential initial intrusion. 
Then, defender should avoid honey links from node $n_{j_0}$ for all stages, i.e., $\gamma_{j_0,w}=0, \forall n_w\in \mathcal{V}$. 
If a honey link from $n_l$ to $n_w$, e.g., the link from U$1$ to H$3$, is not available for all stages due to  segregation, then  $\gamma_{l,w}=0$. 
Since at most one link is allowed, we have the constraint $\sum_{n_l,n_w\in \mathcal{V} }\gamma_{l,w}=1$. 
In this work, we assume that the honeypot policy $\gamma$ is not affected by the realization of the service links at each stage
and thus can interfere with the service links that are not idle as defined in Definition \ref{def:idle}.  If the honeypot $n_w^k$ selected by the policy $\gamma$ is interfering, i.e., not \textit{idle},  
then the defender neither monitors nor filters the inbound network flows to avoid any interference with the normal operation.  

Although we increase the difficulty for the attacker to identify the honeypot by applying it to idle nodes in the network and change its location at every stage, we cannot eliminate the possibility of advanced attackers identifying the honeypot \cite{krawetz2004anti}. 
If the attacker has compromised node $n_i$ before stage $k$ and there is a honey link from node $n_i^k$ to $n_j^k$ at stage $k$, then we assume that the attacker has probability $q_{i,j}\in [0,1]$ to identify the honey link and choose not to access the honeypot. 
If the honeypot is not identified, then the attacker accesses the honeypot and he is detected by the defender. We assume the defender can deter the lateral movement completely after a detection from any single honeypot by patching or resetting all nodes at that stage. 
As stated in Section \ref{subsec:attack model}, the attacker can move simultaneously from all the compromised nodes to multiple nodes through service links that connect them. For example, the attacker at stage $k_0+2$ can compromise H$2$ and H$1$ through the two service links and  may also reach the honeypot if the attacker attempts to compromise H$3$ from U$1$. 
However, we assume that the attacker at a compromised node does not move consecutively through multiple service links (or honey links defined in Section \ref{subsec:cog honeypot} as the attacker cannot distinguish honey links from service ones) in a single stage to remain stealthy. 
Contrary to the persistent lateral movement over a long time period, consecutive attack moves within one stage make it easier for the defender to connect all the indicators of compromise (IoCs) and attribute the attacker. 
Take Fig. \ref{fig:TVG} as an example. Suppose that there are two links, e.g., H$1$ to U$2$ and U$2$ to H$2$ at a stage $k$, where each link can be either a service link or a honey link. If the attacker has only compromised H$1$ among these three nodes, then he only attempts to compromise node U$2$ rather than both U$2$ and H$2$ during stage $k$.  

\subsubsection{Interference, Stealthiness, and Cost of Roaming}
In this section, we define three critical security metrics for a cognitive honeypot to achieve low interference, low cost, and high stealthiness.  
Define $\mathcal{V}_S$ as the set of all the subsets of $\mathcal{V}$. 
Define a series of binary random variables $x^k_{v,w,v'}\in \{0,1\}, v,v'\in \mathcal{V}_S, n^k_w\in \mathcal{V}$, where $x^k_{v,w,v'}=1$ means that there are no direct service links from any node $n_l^k\in v$ to node $n_w^k$ and from $n_w^k$ to $n_l^k\in v'$ at stage $k$. Thus, $\Pr(x^k_{ v,w,v'}=1)=\prod_{n_l^k\in v} (1-\beta_{l,w})\prod_{n_{l'}^k\in v'} (1-\beta_{w,l'})$ represents the probability that the honeypot at $n_w^k$ does not interfere with any service link whose source node is in set $v$ and sink node is in $v'$. 
Then, we can define $ H_{PoI}(\gamma)$ as the probability of interference in Definition \ref{def:PoC}. 
Since the defender can only apply cognitive honeypots to idle nodes, a low probability of interfering can increase efficiency.  
To reduce $ H_{PoI}(\gamma)$, the defender can design $\gamma$ based on  the value of $\beta$, i.e., the frequency/probability of all potential service links. 
\begin{definition}
\label{def:PoC}
The \textbf{probability of interference} (PoI) for any honeypot policy $\gamma$ is 
\begin{equation}
 \begin{split}
 H_{PoI}(\gamma)&:=\sum_{n_h\in \mathcal{V}}\sum_{n_w \in \mathcal{V}\setminus \{n_h\}}\gamma_{h,w} (1-\Pr(x^k_{ \mathcal{V}\setminus \{n_w\},w,\mathcal{V}\setminus \{n_w\}}=1) )\\
 &=\sum_{n_w\in \mathcal{V}} (1-\Pr(x^k_{ \mathcal{V}\setminus \{n_w\},w,\mathcal{V}\setminus \{n_w\}}=1) ) \sum_{n_h\in \mathcal{V}\setminus \{n_w\}}\gamma_{h,w}. 
 \end{split}
 \end{equation}  
\end{definition}

Since the attacker can learn the honeypot policy $\gamma$, the defender prefers the policy to be as random as possible to increase the stealthiness of the honeypot. 
A fully random policy that assigns equal probability to all possible honey links provides forward and backward security; i.e., even if an attacker identifies the honeypot at stage $k$, he cannot use that information to deduce the location of the honeypots in the following and previous stages. 
We use $H_{SL}(\gamma)$, the entropy of $\gamma$ in Definition \ref{def:stealthiness} as a measure for the stealthiness level of the honeypot policy where we define $0\cdot \log 0=0$.  

\begin{definition}
\label{def:stealthiness}
The \textbf{stealthiness level} (SL) for any  $\gamma$ is $H_{SL}(\gamma):=\sum_{n_h,n_w\in \mathcal{V}} \gamma_{h,w} \log (\gamma_{h,w})$. 
\end{definition}

A tradeoff of roaming honeypots hinges on the cost to reconfigure the idle nodes when the defender changes the location of the honeypot and the honey link. 
Define the term $C(\gamma_{h_1,w_1},\gamma_{h_2,w_2}), \allowbreak
 \forall  n_{h_1},n_{h_2},n_{w_1},n_{w_2}\in \mathcal{V}$, as the cost of changing a $(n_{h_1}-n_{w_1})$ honey link  to a $(n_{h_2}-n_{w_2})$ honey link. Note that this cost captures the cost of changing the honeypot location from $w_1$ to $w_2$.  
If only the location change of honeypots incurs a cost, we can let 
$C(\gamma_{h_1,w},\gamma_{h_2,w})=0, \allowbreak \forall h_1\neq h_2, \forall n_w\in \mathcal{V}$,  
without loss of generality. 
We define the cost of roaming  in Definition \ref{def:cost}. 
\begin{definition}
\label{def:cost}
 The \textbf{cost of roaming} (CoR) for any honeypot policy $\gamma$ is
\begin{equation} 
 \begin{split}
 &H_{CoR}(\gamma):= 
 \sum_{n_{h_1}\in \mathcal{V}} \sum_{n_{w_1}\in \mathcal{V}\setminus \{n_{h_1}\}} 
\gamma_{h_1,w_1} (1-\Pr(x^k_{ \mathcal{V}\setminus \{n_{w_1}\},w_1,\mathcal{V}\setminus \{n_{w_1}\}}=1) ) \\
\cdot &\sum_{n_{h_2}\in \mathcal{V}} \sum_{n_{w_2}\in \mathcal{V}\setminus \{h_2\}} \gamma_{h_2,w_2} (1-\Pr(x^k_{ \mathcal{V}\setminus \{n_{w_2}\},w_2,\mathcal{V}\setminus \{n_{w_2}\}}=1) ) \cdot C(\gamma_{h_1,w_1},\gamma_{h_2,w_2})
 \end{split}
 \end{equation} 
 
\end{definition}

\section{Farsighted Vulnerability Mitigation for Long-Term Security}
\label{sec:riskminimization}
Throughout the entire operation of the enterprise network, the defender does not know whether, when, and where the initial intrusion  has happened. 
The defender also cannot know attack paths until a honeypot detects the lateral movement attack. 
Therefore, instead of reactive policies to mitigate attacks that have happened at known stages, we aim at proactive and persistent policies that prepare for the initial intrusion at any stage $k_0$ over a time window of length $\Delta k$. 
That means that the honeypot should roam persistently at all stages according to the policy $\gamma$ to reduce LTV, i.e., the probability that an initial intrusion can reach and compromise the target node within $\Delta k$ stages. 

Given the target node $n_{j_0}\in \mathcal{V}\setminus \mathcal{V}_I$, a subset $v\in \mathcal{V}_S$, and the defender's honeypot policy $\gamma$, we define $g_{j_0}(v,\gamma, \Delta k)$ as the probability that an attacker who has compromised the set of nodes $v$ can compromise the target node $n_{j_0}$ within $\Delta k$ stages. 
Since the initial intrusion happens to a single node $n_{i}\in \mathcal{V}_I$ with probability $\rho_i $ as argued in Section \ref{subsec:attack model},  the {$\Delta k$-stage vulnerability} of the target node $n_{j_0}$ defined in Definition \ref{def:vulnerability} equals $\bar{g}^{\Delta k}_{j_0,\mathcal{V}_I}( \gamma):=\sum_{n_i\in \mathcal{V}_I} \rho_i g_{j_0}(\{n_i\},\gamma, \Delta k)$. In this paper, we refer to $\Delta k$-stage vulnerability as LTV when $\Delta k>1$. 
\begin{definition}[Long-Term Vulnerability]
\label{def:vulnerability}
The \textbf{$\Delta k$-stage vulnerability} of the target node $n_{j_0}$ is the probability that an attacker in the DMZ $\mathcal{V}_I$ can compromise the target node  $n_{j_0}$ within a time window of $\Delta k$ stages. 
\end{definition}

The length of the time window represents the attack's time-effectiveness which is determined by the system setting and the defender's detection efficiency. For example, $\Delta k$ can be the time-to-live (typically on the order of days \cite{purvine2016graph}) for re-authentication to invalidate the PtH attack. 
For another example, 
suppose that the defender can detect and deter the attacker after the initial intrusion yet with a delay due to the high rate of false alarms. If the delay can be contained within $\Delta k_0$ stages, then the defender should choose the honeypot policy to minimize the $\Delta k_0$-stage vulnerability. 
Consider a given threshold $T_0\in [0,1]$, we define the concept of  level-$T_0$ stage-$\Delta k$ security for node $n_{j_0}$ and  honeypot policy $\gamma$ in Definition \ref{def:security}. 
\begin{definition}[Long-Term Security]
\label{def:security}
Policy $\gamma$ achieves \textbf{level-$T_0$ stage-$\Delta k$ security} for node $n_{j_0}$ if the $\Delta k$-stage vulnerability is less than the threshold, i.e., $\bar{g}^{\Delta k}_{j_0,\mathcal{V}_I}( \gamma)\leq T_0$. 
\end{definition}

Finally, we define the defender's decision problem of a cognitive honeypot that can minimize the LTV for the target node with a low PoI, a high SL, and a low CoR in \eqref{CO}. 
The coefficients $ \alpha_{PoI}, \alpha_{SL}, \alpha_{CoR}$ represent the tradeoffs of $\Delta k$-stage vulnerabilities with PoI, SL, and CoR, respectively. 
\begin{equation}
\begin{split}
\label{CO}
 \min_{\gamma }  \quad &  \bar{g}^{\Delta k}_{j_0,\mathcal{V}_I}( \gamma) + \alpha_{PoI}  H_{PoI}(\gamma ) -\alpha_{SL} H_{SL}(\gamma) + \alpha_{CoR} H_{CoR}(\gamma) \\
\text{s.t. } & \sum_{n_h,n_w\in \mathcal{V}} \gamma_{h,w}=1, \\
& \gamma_{h,w}=0, \forall n_h\in \mathcal{V}, n_w\in \mathcal{V}\setminus \mathcal{V}_D.  \\
\end{split}
\end{equation}

\subsection{Imminent Vulnerability}
We first compute the probability that an initial intrusion at node $n_i\in \mathcal{V}_I$ can compromise the target node $n_{j_0}\in \mathcal{V}\setminus \mathcal{V}_I$ within $\Delta k=0$ stages. 
The term $\gamma_{i,w} (1-{q}_{i,w})$ is the \textit{Probability of Immediate  Capture (\textbf{PoIC})}, i.e., the attacker with initial intrusion at node $n_i$ is directly trapped by the honeypot $n_w$.  
Since the attacker does not take consecutive movements in one stage to remain stealthy as stated in Section \ref{subsec:attack model}, $g_{j_0}(\{n_i\},\gamma, 0)$ equals the product of the probability that attacker exploits the service link from $n_i$ to $n_{j_0}$ successfully and the probability that the attacker is not trapped by the honeypot, i.e., $\forall  n_i\in \mathcal{V}_I$, 
\begin{equation}
\begin{split}
\label{eq:noconstraint0}
g_{j_0}(\{n_i\},\gamma, 0) = \beta_{i,j_0} \lambda_{i,j_0}  (1-\sum_{w\neq i,j_0} \gamma_{i,w} (1-{q}_{i,w} ) \Pr(x^k_{ \mathcal{V}\setminus \{n_w\},w,\mathcal{V}\setminus \{n_w\} }=1)  ). 
\end{split}
\end{equation}

\subsection{$\Delta k$-stage Vulnerability}
Define $\mathcal{V}_{i,j_0} \subseteq \mathcal{V}_S$ as the set of all the subsets of $\mathcal{V}\setminus \{n_i,n_{j_0}\}$. 
For each $v\in \mathcal{V}_{i,j_0}$, define $\mathcal{V}^{v}_{i,j_0}$ as the set of all the subsets of $\mathcal{V}\setminus \{n_i,n_{j_0},v\}$. 
Define the shorthand notation 
$ f_{v,u}(\beta,\lambda) := 
\prod_{n_{h_1}\in v}\beta_{i,h_1}\lambda_{i,h_1} \prod_{n_{h_2}\in u}\beta_{i,h_2}(1-\lambda_{i,h_2} ) \prod_{n_{h_3}\in \mathcal{V}\setminus \{n_i,n_{j_0},v,u\}}(1-\beta_{i,h_3})$ as the \textit{probability of partial compromise}, i.e., the attacker with initial intrusion at node $n_i$ has compromised the service links from $n_i$ to all nodes in set $v\in \mathcal{V}_{i,j_0}$, yet fails to compromise the remaining service links from $n_i$ to all nodes in set $u \in  \mathcal{V}^{v}_{i,j_0}$. 
We can compute $g_{j_0}(\{n_i\},\gamma, \Delta k)$ based on the following induction, i.e., 
\begin{equation}
\begin{split}
\label{eq:noconstraint>1}
&  g_{j_0}(\{n_i\},\gamma, \Delta k)  =  g_{j_0}(\{n_i\},\gamma, 0)  +(1-\beta_{i,j_0} \lambda_{i,j_0} )
\sum_{v \in \mathcal{V}_{i,j_0}} \sum_{u \in \mathcal{V}^{v}_{i,j_0} }  f_{v,u}(\beta,\lambda) (1- 
\\
&
\sum_{n_w \in \mathcal{V}\setminus \{n_i,v,u\}} \gamma_{i,w} (1-{q}_{i,w} ) \Pr(x^k_{ \mathcal{V}\setminus \{n_i,n_w\},w,\mathcal{V}\setminus \{n_w\}}=1)   ) g_{j_0}(\{n_i\} \cup v,\gamma,\Delta k-1). 
\end{split}
\end{equation}

 \subsection{Curse of Multiple Attack Paths and Two Sub-Optimal Honeypot Policies
 } 
For a given $\gamma$, we can write out the explicit form of $g_{j_0}(\{n_i\} \cup v,\gamma,\Delta k-1)$ for all $\Delta k\in \mathbb{Z}^+$ as in \eqref{eq:noconstraint0} and \eqref{eq:noconstraint>1}. However, the complexity increases dramatically with the cardinality of set $v$ due to the \textit{curse of multiple attack paths}; i.e., the event that the attacker can compromise target node $n_{j_0}$ within $\Delta k$ stages from node $n_i$ is not independent of the event that the attacker can achieve the same compromise from node $n_h\neq n_i$. 
Thus, we use the union bound 
\begin{equation*}
\begin{split}
& g_{j_0}(\{n_i\}\cup v ,\gamma,\Delta k)\geq \max_{n_j\in \{n_i\}\cup v } g_{j_0}(\{n_j\},\gamma,\Delta k),  \\
&  g_{j_0}(\{n_i\}\cup v ,\gamma,\Delta k)\leq \min (1,\sum_{n_j\in \{n_i\}\cup v } g_{j_0}(\{n_j\},\gamma,\Delta k)),  
\end{split}
\end{equation*}
to simplify the computation and provide an upper bound and a lower bound for $g_{j_0}(\{n_i\}\cup v ,\gamma, \Delta k), v\neq \emptyset, \forall \Delta k\in \mathbb{Z}^+$, in \eqref{eq:lower} and \eqref{eq:upper}, respectively. 
\begin{equation}
\begin{split}
\label{eq:lower}
&  g^{lower}_{j_0}(\{n_i\},\gamma, \Delta k)  =  g_{j_0}(\{n_i\},\gamma, 0)  +(1-\beta_{i,j_0} \lambda_{i,j_0} )
\sum_{v \in \mathcal{V}_{i,j_0}} \sum_{u \in \mathcal{V}^{v}_{i,j_0} }  f_{v,u}(\beta,\lambda) (1- 
\\
&
\sum_{n_w \in \mathcal{V}\setminus \{n_i,v,u\}} \gamma_{i,w} (1-{q}_{i,w} ) \Pr(x^k_{ \mathcal{V}\setminus \{n_i,n_w\},w,\mathcal{V}\setminus \{n_w\}}=1)   ) \max_{n_j\in \{n_i\}\cup v } g^{lower}_{j_0}(\{n_j\},\gamma,\Delta k-1). 
\end{split}
\end{equation}
\begin{equation}
\begin{split}
\label{eq:upper}
  g^{upper}_{j_0}(\{n_i\},\gamma, \Delta k)  = & g_{j_0}(\{n_i\},\gamma, 0)  +(1-\beta_{i,j_0} \lambda_{i,j_0} )
\sum_{v \in \mathcal{V}_{i,j_0}} \sum_{u \in \mathcal{V}^{v}_{i,j_0} }  f_{v,u}(\beta,\lambda) \\
&
\cdot (1- 
\sum_{n_w \in \mathcal{V}\setminus \{n_i,v,u\}}  \gamma_{i,w} (1-{q}_{i,w} )    \Pr(x^k_{ \mathcal{V}\setminus \{n_i,n_w\},w,\mathcal{V}\setminus \{n_w\}}=1)) \\
&\cdot \min (1,\sum_{n_j\in \{n_i\}\cup v } g^{upper}_{j_0}(\{n_j\},\gamma,\Delta k-1)). 
\end{split}
\end{equation}
The initial condition at $\Delta k=0$ is $ g^{lower}_{j_0}(\{n_j\},\gamma,0)=g^{upper}_{j_0}(\{n_j\},\gamma,0)= g_{j_0}(\{n_j\},\gamma,0), \allowbreak 
\forall n_j\in \{n_i\}\cup v $. 
Define  $\bar{g}^{\Delta k, lower}_{j_0,\mathcal{V}_I}( \gamma):=\sum_{n_i\in \mathcal{V}_I} \rho_i g^{lower}_{j_0}(\{n_i\},\gamma, \Delta k)$ and $\bar{g}^{\Delta k, upper}_{j_0,\mathcal{V}_I}( \gamma):=\sum_{n_i\in \mathcal{V}_I} \rho_i g^{upper}_{j_0}(\{n_i\},\gamma, \Delta k)$ as the lower and upper bounds of the {$\Delta k$-stage vulnerability} of the target node $n_{j_0}$ under any given policy $\gamma$, respectively. Then, replacing $\bar{g}^{\Delta k}_{j_0,\mathcal{V}_I}( \gamma)$ in \eqref{CO} with $\bar{g}^{\Delta k, lower}_{j_0,\mathcal{V}_I}( \gamma)$ and $\bar{g}^{\Delta k, upper}_{j_0,\mathcal{V}_I}( \gamma)$, we obtain the optimal risky and conservative honeypot policy $\gamma^{*,risky}$ and $\gamma^{*,cons}$, respectively. 
Both sub-optimal honeypot policies approximate the optimal policy that is hard to compute explicitly. A risky defender can choose $\gamma^{*,risky}$ to minimize the lower bound of LTV while a conservative defender can choose $\gamma^{*,cons}$ to minimize the upper bound.  

We propose the following iterative algorithm to compute these two honeypot policies. We use $\gamma^{*,risky}$ as an example and $\gamma^{*,cons}$ can be computed in the same fashion. 
At iteration $t\in \mathbb{Z}^+_0$, we consider any feasible honeypot policy $\gamma^t$ and compute $g^{lower}_{j_0}(\{n_i\},\gamma^t, \Delta k'), \forall n_i\in \mathcal{V}_I, \forall \Delta k'\in \{1,\cdots,\Delta k\}$, via \eqref{eq:lower}. 
Then, we solve  \eqref{CO} by replacing $\bar{g}^{\Delta k}_{j_0,\mathcal{V}_I}( \gamma^t)$ with $\bar{g}^{\Delta k, lower}_{j_0,\mathcal{V}_I}( \gamma^t)$ and plugging in $g^{lower}_{j_0}(\{n_i\},\gamma^t, \Delta k), \forall n_i\in \mathcal{V}_I$, as constants. 
Since  $\bar{g}^{\Delta k, lower}_{j_0,\mathcal{V}_I}( \gamma^t), 
\allowbreak
H_{PoI}(\gamma^t), \allowbreak
H_{CoR}(\gamma^t)$ are all linear with respect to $\gamma^t$, 
the objective function of the constrained optimization in \eqref{CO} is a linear function of $\gamma^t$ plus the entropy regularization $H_{SL}(\gamma^t )$.  
Then, we can solve the constrained optimization in closed form and update the honeypot policy from $\gamma^{t}$ to $\gamma^{t+1}$. 
Given a small error threshold $\epsilon>0$, the above iteration process can be repeated until there exists a $T_1\in \mathbb{Z}^+_0$ such that a proper matrix norm is less than the error threshold, i.e., $||\gamma^{T_1+1}-\gamma^{T_1}||\leq \epsilon$. 
Then, we can output $\gamma^{T_1+1}$ as the optimal risky honeypot policy $\gamma^{*,risky}$. 

\begin{algorithm}[h]
\SetAlgoLined
 Initialization $\mathcal{V}_I, n_{j_0}\in \mathcal{V}\setminus \mathcal{V}_I, \Delta k\in \mathbb{Z}^+,\epsilon>0, \gamma^0$,$t=0$\;  
 
 \While{ $||\gamma^{t+1}-\gamma^{t}|| > \epsilon$ }{
  \For{ $\Delta k'=1,\cdots,\Delta k$ }
 {
 			\For{ $i\in \mathcal{V}_I$ }
 			{
 Compute $g^{lower}_{j_0}(\{n_i\},\gamma^t, \Delta k')$ via \eqref{eq:lower}\;
			 }
  }
Replace $\bar{g}^{\Delta k}_{j_0,\mathcal{V}_I}( \gamma^t)$ with $\bar{g}^{\Delta k, lower}_{j_0,\mathcal{V}_I}( \gamma^t)$ and plug in $g^{lower}_{j_0}(\{n_i\},\gamma^t, \Delta k), \forall n_i\in \mathcal{V}_I$\;
Obtain $\gamma^{t+1}$ as the solution of \eqref{CO}\;
  \uIf{ $||\gamma^{t+1}-\gamma^{t}||\leq \epsilon$ }{
  $T_1=t$\;
\textbf{Terminate}\
  }
   $t:=t+1$\;
}
  \textbf{Output}  $\gamma^{*,risky}=\gamma^{T_1+1}$. 
 \caption{Optimal Risky (and Conservative) Honeypot Policy}
\end{algorithm}

 \subsection{LTV Analysis under two Heuristic Policies} 
In this section, we consider the scenario where the initial intrusion set $\mathcal{V}_I=\{n_i\}$ contains only one node $n_i$, i.e., the attacker cannot compromise other nodes directly from the external network at stage $k_0$. 
Then, a reasonable heuristic policy is to set up the honeypot at a fixed node $n_{w_0}\in \mathcal{V}\setminus \{ n_i,n_{j_0}\}$ whenever the node is idle and also a direct honey link from $n_i$ to $n_{w_0}$. 
We refer to these deterministic policies with $\gamma_{i,w_0}=1$ as the direct policies in Section \ref{subsection:indirect}. 

In the second scenario, the defender further segregates node $n_i$ from the external network to form a \textit{air gap} so that she chooses to apply no direct honey links from $n_i$ to any honeypot at all stages, i.e.,$\gamma_{i,w}=0, \forall n_w \in \mathcal{V}$. 
However, advanced attacks, such as Stuxnet, can cross the air gap by an infected USB flash drive to accomplish the initial intrusion to the air-gap node $n_i$ and then move laterally to the entire network $\mathcal{V}$. 
Although the defender mistakenly sets up no honey links from $n_i$ to the honeypot at all stages, other indirect honey links with source nodes other than $n_i$ may also detect the lateral movement in $\Delta k$ stages. 
Unlike the deterministic direct policies, we refer to these stochastic policies with $\gamma_{i,w}=0, \forall n_w \in \mathcal{V}$, as the indirect policies in Section \ref{subsection:direct}.  

Since the defender may adopt these heuristic policies in the listed scenarios, 
this section aims to analyze the LTV under the direct and indirect policies to answer the following security questions. 
How effective is the lateral movement for a different length of duration time under heuristic policies? What are the limit and the bounds of the vulnerability when the window length goes to infinity? 
How much additional vulnerability is introduced by adopting improper indirect policies rather than the direct policies? 
How to change the value of parameters, such as $\beta$ and $\lambda$, to reduce LTV if they are designable?


 
 \subsubsection{Indirect Honeypot Policies}
 \label{subsection:indirect}
 Since the defender overestimates the effectiveness of air gap and chooses the improper honeypot policies that $\gamma_{i,w}=0, \forall n_w \in \mathcal{V}$, the vulnerability of any target node $n_{j_0}$ is non-decreasing with the length of the time window  as shown in Proposition \ref{proposition:Increasing}. 
 
\begin{proposition}[Non-Decreasing Vulnerability over Stages]
\label{proposition:Increasing}
If the PoIC is zero, i.e., $\gamma_{i,w} (1-{q}_{i,w} ) =0, \forall n_w\in \mathcal{V}$, then the vulnerability $g_{j_0}(\{n_i\},\gamma,\Delta k)\in [0,1]$ is an non-decreasing function regarding $\Delta k$ for all target node $n_{j_0}\in \mathcal{V}\setminus \mathcal{V}_I, n_i\in \mathcal{V}_I$. 
The value of $g_{j_0}(\{n_i\},\gamma,\Delta k)$ does not increase to $1$ as $\Delta k$ increases to infinity if and only if  $\beta_{i,j_0} \lambda_{i,j_0}= 0$ and $ g_{j_0}(\{n_i\} \cup v,\gamma,\Delta k-1)= g_{j_0}(\{n_i\},\gamma,\Delta k-1), \forall v\in \mathcal{V}_S, \forall \Delta k\in \mathbb{Z}^+$. 
\end{proposition}
\begin{proof}
If $\gamma_{i,w} (1-{q}_{i,w} ) =0, \forall n_w\in \mathcal{V}$, we can use the facts that $ g_{j_0}(\{n_i\} \cup v,\gamma,\Delta k-1)\geq  g_{j_0}(\{n_i\},\gamma,\Delta k-1), \forall \gamma, n_{j_0}\in \mathcal{V}, n_i\in \mathcal{V}_I, \Delta k\geq 0, \forall v\in \mathcal{V}_S$, and 
$
\sum_{v \in \mathcal{V}_{i,j_0}} \sum_{u \in \mathcal{V}^{v}_{i,j_0} }  f_{v,u}(\beta,\lambda) \equiv1, \forall \beta, \lambda, 
$
to obtain $ g_{j_0}(\{n_i\},\gamma, \Delta k) $ as 
\begin{equation}
\begin{split}
\label{eq:gammaIS1}
&      \beta_{i,j_0} \lambda_{i,j_0}  +
(1-  \beta_{i,j_0} \lambda_{i,j_0} )\sum_{v \in \mathcal{V}_{i,j_0}} \sum_{u \in \mathcal{V}^{v}_{i,j_0} }  f_{v,u}(\beta,\lambda) g_{j_0}(\{n_i\} \cup v,\gamma,\Delta k-1) \\
& \geq 
\beta_{i,j_0} \lambda_{i,j_0}  + 
(1-  \beta_{i,j_0} \lambda_{i,j_0} ) g_{j_0}(\{n_i\},\gamma,\Delta k-1)\geq  g_{j_0}(\{n_i\},\gamma,\Delta k-1), 
\end{split}
\end{equation}
 for all $\Delta k\in \mathbb{Z}^+$.  The inequality is an equality if and only if  $\beta_{i,j_0} \lambda_{i,j_0}= 0$ and  $ g_{j_0}(\{n_i\} \cup v,\gamma,\Delta k-1)= g_{j_0}(\{n_i\},\gamma,\Delta k-1), \forall v\in \mathcal{V}_S, \forall \Delta k\in \mathbb{Z}^+$.  
 \qed
\end{proof}
The equation $ g_{j_0}(\{n_i\} \cup v,\gamma,\Delta k-1)= g_{j_0}(\{n_i\},\gamma,\Delta k-1), \forall v\in \mathcal{V}_S, \forall \Delta k\in \mathbb{Z}^+$, holds only under very unlikely conditions such as there is only one node in the network, i.e., $N=1$ or service links occur only from node $n_i$, i.e., $\lambda_{i',j}=0, \forall i'\neq i, \forall n_j\in \mathcal{V}$. Thus, except for these rare special cases, the vulnerability $g_{j_0}(\{n_i\},\gamma,\Delta k)$ always increases to the maximum value of $1$ under indirect policies.

\begin{remark}
Proposition \ref{proposition:Increasing} shows that without a proper mitigation strategy, e.g., no direct honey link from the initial intrusion node to the honeypot, 
the vulnerability of a target node never decreases over stages. Moreover, except from rare special cases, the target node will be compromised with probability $1$ as time goes to infinity. 
\end{remark}

Proposition \ref{proposition:Increasing} demonstrates the disadvantaged position of the defender against persistent lateral movement without proper honeypot policies. 
Under these disadvantageous situations, the defender may need alternative security measures to mitigate the LTV. 
For example, the defender may reduce the arrival frequency of the service link from $n_{j_1}$ to $n_{j_2}$, i.e., $\beta_{j_1,j_2}$, to delay lateral movement at the expenses of operational efficiency. 
Also, the defender may attempt to reduce  the probability of a successful compromise from node  $n_{j_1}$ to $n_{j_2}$, i.e., $\lambda_{j_1,j_2}$, by filtering the service link from  $n_{j_1}$ to $n_{j_2}$ with more stringent rules or demotivate the attacker to initiate the link compromise by disguising the service link as a honey link.  
In the rest of this subsection, we briefly investigate the influence of $\beta$ and $\lambda$ on the $\Delta k$-stage vulnerability under indirect policies. 

The probability of no direct link from the initial intrusion node $n_i$ to target $n_{j_0}$, i.e., $1-\beta_{i,j_0} \lambda_{i,j_0}$, and the probability that the attacker at node $n_i$ is demotivated to or fails to compromise the service links from node $n_i$, i.e., $\sum_{u \in \mathcal{V}^{\emptyset}_{i,j_0} }  f_{\emptyset,u} (\beta,\lambda)$, defines the \textit{Probability of Movement Deterrence (\textbf{PoMD})} $r:=(1-\beta_{i,j_0} \lambda_{i,j_0}) \sum_{u \in \mathcal{V}^{\emptyset}_{i,j_0} }  f_{\emptyset,u} (\beta,\lambda)$. 
In \eqref{eq:gammaIS1} where the PoIC is $0$, i.e., $\gamma_{i,w} (1-{q}_{i,w} ) =0, \forall n_w\in \mathcal{V}$, we can upper bound the term $g_{j_0}(\{n_i\} \cup v,\gamma,\Delta k-1)$ by $1$ for all $v \neq \emptyset$, which leads to 
\begin{equation}
\begin{split}
\label{eq:induction}
 g_{j_0}(\{n_i\},\gamma, \Delta k) &  =  (1-r)\cdot g_{j_0}(\{n_i\}\cup v ,\gamma,\Delta k-1) +  r \cdot g_{j_0}(\{n_i\} ,\gamma,\Delta k-1)\\
 & \leq  (1-r)+  r \cdot g_{j_0}(\{n_i\} ,\gamma,\Delta k-1)\\
 & =1-r^{\Delta k}+r^{\Delta k} g_{j_0}(\{n_i\},\gamma, 0) = 1-r^{\Delta k} (1- \beta_{i,j_0} \lambda_{i,j_0} ), 
\end{split}
\end{equation}
where the final line results from solving the  first-order linear difference equation iteratively by $\Delta k-1$ times. 

Equation \eqref{eq:induction} shows that the upper bound of LTV increases exponentially concerning the duration of lateral movement $\Delta k$ yet decreases in a
polynomial growth rate as PoMD increases. 
Note that letting PoMD be $1$ can completely deter lateral movement and achieve zero LTV for any $\Delta k\in \mathbb{Z}^+$. However, it is challenging to attain it as it requires the attacker do not succeed from $n_i$ to any node $n_j$ with probability $1$, i.e., $\lambda_{i,j}=0, \forall n_j\in \mathcal{V}$. 
Since increasing PoMD incurs a higher cost (e.g., reducing the compromise rate $\lambda$) and lower operational efficiency (e.g., reducing the frequency of service links $\beta$), we aim to find the minimum PoMD to mitigate LTV even when the duration of lateral movement $\Delta k$ goes to infinity. 
In Proposition \ref{proposition:doubleconverge}, we characterize the critical \textit{Threshold of Compromisability (\textbf{ToC})} $T^{ToC}_m:=1-m/\Delta k$ for a positive $m\ll \Delta k$ to guarantee a level-$(\beta_{i,j_0} \lambda_{i,j_0})$, stage-$\infty$ security defined in Definition \ref{def:security}. The proof follows directly from a limit analysis based on \eqref{eq:induction}. 

\begin{proposition}[ToC ]
\label{proposition:doubleconverge}
Consider the scenario where $\gamma_{i,w} (1-{q}_{i,w} )=0, \forall n_w\in \mathcal{V}$, and $r$ as a function of $\Delta k$ has the form  $r =1-m  \Delta k^{-n}$, where $n,m\in \mathbb{R}^+$ and $m\ll \Delta k$. 
\begin{itemize}
\item[(1).] If $(1-r)/m$ is of the same order with $1/\Delta k$, i.e., $n=1$, then the limit of the upper bound $\lim_{\Delta k \rightarrow \infty} 1-r^{\Delta k} (1- \beta_{i,j_0} \lambda_{i,j_0} )$ is a constant $1-e^{-m}(1- \beta_{i,j_0} \lambda_{i,j_0})$. 
\item[(2).] If $(1-r)/m$  is of higher order, i.e., $n>1$, then the limit of the upper bound  is $g_{j_0}(\{n_i\},\gamma, 0)= \beta_{i,j_0} \lambda_{i,j_0} $. If $\beta_{i,j_0} \lambda_{i,j_0}=0$, zero LTV is achieved $g_{j_0}(\{n_i\},\gamma, \infty )= 0$. 
\item[(3).] If  $(1-r)/m$  is  of lower order, i.e., $n<1$, then  the limit of the upper bound  is $1$.  
\end{itemize}
\end{proposition}

Based on the fact that $1-e^{-m}(1- \beta_{i,j_0} \lambda_{i,j_0} )\geq \beta_{i,j_0} \lambda_{i,j_0}$ where the equality holds if and only if $\beta_{i,j_0} \lambda_{i,j_0}=1$,  we can conclude that if $r\geq T^{ToC}_m$ for a positive $m\ll \Delta k$, then the $\infty$-stage vulnerability of target node $n_{j_0}$ is upper bounded by $\beta_{i,j_0} \lambda_{i,j_0}$ and thus achieves the  level-$(\beta_{i,j_0} \lambda_{i,j_0})$, stage-$\infty$ security as defined in Definition \ref{def:security}. 
Note that if the target node is segregated from nodes in DMZ $\mathcal{V}_I$ for the sake of security, then there is no direct service link from node $n_i$ to the target node $n_{j_0}$ and  $\beta_{i,j_0} \lambda_{i,j_0}=0$. 
In that case, the target node $n_{j_0}$ can achieve a zero vulnerability for an infinite duration of lateral movement, i.e., $g_{j_0}(\{n_i\},\gamma, \infty )= 0$, because the upper bound is $0$ and LTV is always non-negative. 


\subsubsection{Direct Honeypot Policies}
 \label{subsection:direct}
For the direct policies $\gamma_{i,w_0}=1, n_{w_0}\in \mathcal{V}\setminus \{ n_i,n_{j_0}\}$, we obtain the corresponding $\Delta k$-stage vulnerability and an explicit lower bound in \eqref{eq:directpolicy} based on \eqref{eq:noconstraint>1} by using the inequality $g_{j_0}(\{n_i\}\cup v ,\gamma,\Delta k-1) \geq g_{j_0}(\{n_i\} ,\gamma,\Delta k-1) $. 
Define shorthand notations $k_1:=\prod_{l\neq w_0} (1-\beta_{l,w_0})(1-\beta_{w_0,l})(1-q_{i,w_0})\in [0,1]$ and $k_2:=\allowbreak 
\sum_{v \in \mathcal{V}_{i,j_0}\setminus \{n_{w_0}\}} \allowbreak 
\sum_{u \in \mathcal{V}^{v}_{i,j_0}\setminus \{w_0\} }  \allowbreak
f_{v,u}(\beta,\lambda) \leq \sum_{v \in \mathcal{V}_{i,j_0}} \sum_{u \in \mathcal{V}^{v}_{i,j_0} }  f_{v,u}(\beta,\lambda) =1 $. 
Note that $k_1=0$ is a very restrictive condition as it requires that the honeypot $n_{w_0}$ is not interfering, i.e., node  $n_{w_0}$ is \textit{idle} and  the attacker never identify the honey link from $n_i$ to $n_{w_0}$, i.e., $q_{i,w_0}=0$. 

\begin{equation}
 \begin{split}
 \label{eq:directpolicy}
 &  g_{j_0}(\{n_i\},\gamma, \Delta k)  = \beta_{i,j_0} \lambda_{i,j_0}  [1-\prod_{l\neq w_0} (1-\beta_{l,w_0})(1-\beta_{w_0,l})(1-q_{i,w_0})] + \\
 &  (1-\beta_{i,j_0} \lambda_{i,j_0} )
 [ \sum_{v \in \mathcal{V}_{i,j_0}} \sum_{u \in \mathcal{V}^{v}_{i,j_0} }  f_{v,u}(\beta,\lambda) g_{j_0}(\{n_i\} \cup v,\gamma,\Delta k-1) 
-\sum_{v \in \mathcal{V}_{i,j_0}\setminus \{n_{w_0}\}} \sum_{u \in \mathcal{V}^{v}_{i,j_0}\setminus \{n_{w_0}\} }   \\
&
f_{v,u}(\beta,\lambda) 
\cdot \prod_{l\neq i,w_0} (1-\beta_{l,w_0})\prod_{l'\neq w_0}(1-\beta_{w_0,l'})(1-q_{i,w_0})  g_{j_0}(\{n_i\} \cup v,\gamma,\Delta k-1) ]\\
&  \geq
\beta_{i,j_0} \lambda_{i,j_0} (1-k_1) + (1-\beta_{i,j_0} \lambda_{i,j_0}) [1- k_1 k_2 (1-\beta_{i,w_0})]  g_{j_0}(\{n_i\},\gamma,\Delta k-1).     
 \end{split}
 \end{equation} 
Define a shorthand notation $r_2:=(1-\beta_{i,j_0} \lambda_{i,j_0}) [1- k_1 k_2 (1-\beta_{i,w_0})] $, we can solve the linear  difference equation  in the final step of \eqref{eq:directpolicy} to obtain an lower bound, i.e.,  $g_{j_0}(\{n_i\},\gamma, \Delta k) \geq T_2^{lower,1}:=\beta_{i,j_0} \lambda_{i,j_0} (1-k_1) \frac{1-(r_2)^{\Delta k+1}}{1-r_2}$ for all $\Delta k\in \mathbb{Z}^+$. 
According to the first equality in \eqref{eq:directpolicy}, we also obtain an upper bound $T_2^{upper}$ for $ g_{j_0}(\{n_i\},\gamma, \Delta k), \forall \Delta k\in \mathbb{Z}^+$, in Lemma \ref{lemma:upperboundedG} by using the inequality $ g_{j_0}(\{n_i\}\cup v,\gamma, \Delta k)\leq 1, \forall v\in \mathcal{V}_{i,j_0}$\footnote{Since we can compute $g_{j_0}(\{n_i\} \cup v,\gamma,\Delta k-1)$ explicitly when $v$ is empty, we can obtain a tighter upper bound by using the inequality $ g_{j_0}(\{n_i\}\cup v,\gamma, \Delta k)\leq 1, \forall v\in \mathcal{V}_{i,j_0}\setminus \emptyset$.}. 
The bound $T_2^{upper}<1$ is non-trivial if $\beta_{i,j_0} \lambda_{i,j_0}\neq 0,\beta_{i,j_0} \lambda_{i,j_0}\neq 1$, and $k_1 k_2(1-\beta_{i,w_0})\neq 0$. 
\begin{lemma}
\label{lemma:upperboundedG}
If $\gamma_{i,w_0}=1, w_0\neq i,j_0$, then $ g_{j_0}(\{n_i\},\gamma, \Delta k) $ is lower and upper bounded by $T_2^{lower,1}$ and $T_2^{upper}:=
1-\beta_{i,j_0} \lambda_{i,j_0} k_1- (1-\beta_{i,j_0} \lambda_{i,j_0})k_1 k_2(1-\beta_{i,w_0})
\in [0,1]$ for all $\Delta k\in \mathbb{Z}^+$, respectively. 
\end{lemma}
 Lemma \ref{lemma:upperboundedG} shows that if the defender applies a direct honeypot from $n_i$ in a deterministic fashion, then the $\Delta k$-stage vulnerability is always upper bounded. 
However, these direct policies cannot reduce the $\infty$-stage vulnerability to zero as shown in Proposition \ref{proposition:boundedG}. 

\begin{proposition}[Vulnerability Residue]
\label{proposition:boundedG}
If $\beta_{i,j_0} \lambda_{i,j_0}\neq 0$ and $\gamma_{i,w_0}=1, w_0\neq i,j_0$, then  
\begin{itemize}
\item[(1).] The term $T_2^{lower,2}:=\frac{\beta_{i,j_0} \lambda_{i,j_0}(1-k_1)}{(1-\beta_{i,j_0} \lambda_{i,j_0})k_1 k_2 (1-\beta_{i,w_0})+\beta_{i,j_0} \lambda_{i,j_0}}\in [0,1)$ is strictly less than $1$. 
\item[(2).] If $g_{j_0}(\{n_i\},\gamma, \Delta k-1)<T_2^{lower,2}$, then $g_{j_0}(\{n_i\},\gamma, \Delta k)> g_{j_0}(\{n_i\},\gamma, \Delta k-1)$. 
\item[(3).] The term $\lim_{\Delta k\rightarrow \infty} g_{j_0}(\{n_i\},\gamma, \Delta k) $ is lower bounded by $ \max (T_2^{lower,1},T_2^{lower,2})$. 
\end{itemize} 
\end{proposition}

\begin{proof}
Based on the inequality in \eqref{eq:directpolicy}, we obtain that if $g_{j_0}(\{n_i\},\gamma, \Delta k-1)< T_2^{lower,2}$, then $g_{j_0}(\{n_i\},\gamma, \Delta k)> g_{j_0}(\{n_i\},\gamma, \Delta k-1)$. 
Since the above is true for all $\Delta k\in \mathbb{Z}^+$, we know that  the  $\Delta k$-stage vulnerability increases with $\Delta k$ strictly until it has reach $T_2^{lower,2}$. 
If  $\beta_{i,j_0} \lambda_{i,j_0}\neq 0$ and $k_1\neq 1$, then $T_2^{lower,2}>0$ is a non-trivial lower bound. The other lower bound $T_2^{lower,1}$ comes from Lemma \ref{lemma:upperboundedG}. 
\qed
\end{proof} 


\begin{remark}
Proposition \ref{proposition:boundedG} defines a \textit{vulnerability residue} $T^{VR}:=\max (T_2^{lower,1},T_2^{lower,2})$ under direct honeypot policies. A nonzero $T^{VR}$ characterizes the limitation of security policies against lateral movement attacks, i.e., LTV cannot be reduced to $0$ as $\Delta k\rightarrow \infty$. 
\end{remark}


\section{Conclusion}
\label{sec:conclusion}
The stealthy and persistent lateral movement of APTs poses a severe security challenge to enterprise networks. 
Since APT attackers can remain undetected in compromised nodes for a long time, a network that is secure at any separate time may become insecure if the times and the spatial locations are considered holistically. Therefore, the defender needs to reduce the LTV of valuable assets. 
Honeypots, as a promising deceptive defense method, can detect lateral movement attacks at their early stages. 
Since advanced attackers, such as APTs, can identify the honeypots located at fixed machines that are segregated from the production system, we propose a cognitive honeypot mechanism which reconfigures idle production nodes as honeypot at different stages based on the probability of service links and successful compromise. 
The time-expanded network is used to model the time of the random service occurrence and the adversarial compromise explicitly. 
Besides the main objective of reducing the target node's LTV, we also consider the level of stealthiness, the probability of interference, and the cost of roaming as three tradeoffs. 
To reduce the computation complexity caused by the curse of multiple attack paths, we propose an iterative algorithm and approximate the vulnerability with the union bound. 
The analysis of the LTV under two heuristic honeypot policies illustrates that without proper mitigation strategies, vulnerability never decreases over stages and the target node is doom to be compromised given sufficient stages of adversarial lateral movement. 
Moreover, even under the improved honeypot strategies, a \textit{vulnerability residue} exists. Thus, LTV cannot be reduced to $0$ and perfect security does not exist.
Besides honeypot policies, we investigate the influence of the frequency of service links and the probability of successful compromise on LTV and characterize a critical threshold to achieve \textit{long-term security}. The target node can achieve zero vulnerability under infinite stages of lateral movement by a modification of the  parameters $\beta,\lambda$ to make PoMD not less than the ToC.

\bibliographystyle{IEEEtran}
\bibliography{TVGforsecurity}

\begin{thebibliography}{10}
\providecommand{\url}[1]{#1}
\csname url@samestyle\endcsname
\providecommand{\newblock}{\relax}
\providecommand{\bibinfo}[2]{#2}
\providecommand{\BIBentrySTDinterwordspacing}{\spaceskip=0pt\relax}
\providecommand{\BIBentryALTinterwordstretchfactor}{4}
\providecommand{\BIBentryALTinterwordspacing}{\spaceskip=\fontdimen2\font plus
\BIBentryALTinterwordstretchfactor\fontdimen3\font minus
  \fontdimen4\font\relax}
\providecommand{\BIBforeignlanguage}[2]{{%
\expandafter\ifx\csname l@#1\endcsname\relax
\typeout{** WARNING: IEEEtran.bst: No hyphenation pattern has been}%
\typeout{** loaded for the language `#1'. Using the pattern for}%
\typeout{** the default language instead.}%
\else
\language=\csname l@#1\endcsname
\fi
#2}}
\providecommand{\BIBdecl}{\relax}
\BIBdecl

\bibitem{ATTACK}
\BIBentryALTinterwordspacing
T.~M. Corporation. (2020) Enterprise matrix. [Online]. Available:
  \url{https://attack.mitre.org/matrices/enterprise/}
\BIBentrySTDinterwordspacing

\bibitem{zhu2018multi}
Q.~Zhu and S.~Rass, ``On multi-phase and multi-stage game-theoretic modeling of
  advanced persistent threats,'' \emph{IEEE Access}, vol.~6, pp.
  13\,958--13\,971, 2018.

\bibitem{APT27}
\BIBentryALTinterwordspacing
D.~Legezo, \emph{LuckyMouse hits national data center to organize country-level
  waterholing campaign}, June 13, 2018. [Online]. Available:
  \url{https://securelist.com/luckymouse-hits-national-data-center/86083/}
\BIBentrySTDinterwordspacing

\bibitem{spitzner2003honeypots}
L.~Spitzner, \emph{Honeypots: tracking hackers}.\hskip 1em plus 0.5em minus
  0.4em\relax Addison-Wesley Reading, 2003, vol.~1.

\bibitem{krawetz2004anti}
N.~Krawetz, ``Anti-honeypot technology,'' \emph{IEEE Security \& Privacy},
  vol.~2, no.~1, pp. 76--79, 2004.

\bibitem{mitola1999cognitive}
J.~Mitola and G.~Q. Maguire, ``Cognitive radio: making software radios more
  personal,'' \emph{IEEE personal communications}, vol.~6, no.~4, pp. 13--18,
  1999.

\bibitem{khattab2004roaming}
S.~M. Khattab, C.~Sangpachatanaruk, D.~Moss{\'e}, R.~Melhem, and T.~Znati,
  ``Roaming honeypots for mitigating service-level denial-of-service attacks,''
  in \emph{24th International Conference on Distributed Computing Systems,
  2004. Proceedings.}\hskip 1em plus 0.5em minus 0.4em\relax IEEE, 2004, pp.
  328--337.

\bibitem{casteigts2012time}
A.~Casteigts, P.~Flocchini, W.~Quattrociocchi, and N.~Santoro, ``Time-varying
  graphs and dynamic networks,'' \emph{International Journal of Parallel,
  Emergent and Distributed Systems}, vol.~27, no.~5, pp. 387--408, 2012.

\bibitem{liu2018latte}
Q.~Liu, J.~W. Stokes, R.~Mead, T.~Burrell, I.~Hellen, J.~Lambert, A.~Marochko,
  and W.~Cui, ``Latte: Large-scale lateral movement detection,'' in
  \emph{MILCOM 2018-2018 IEEE Military Communications Conference
  (MILCOM)}.\hskip 1em plus 0.5em minus 0.4em\relax IEEE, 2018, pp. 1--6.

\bibitem{tian2019real}
Z.~Tian, W.~Shi, Y.~Wang, C.~Zhu, X.~Du, S.~Su, Y.~Sun, and N.~Guizani,
  ``Real-time lateral movement detection based on evidence reasoning network
  for edge computing environment,'' \emph{IEEE Transactions on Industrial
  Informatics}, vol.~15, no.~7, pp. 4285--4294, 2019.

\bibitem{lah2018proposed}
A.~A.~A. Lah, R.~A. Dziyauddin, and M.~H. Azmi, ``Proposed framework for
  network lateral movement detection based on user risk scoring in siem,'' in
  \emph{2018 2nd International Conference on Telematics and Future Generation
  Networks (TAFGEN)}.\hskip 1em plus 0.5em minus 0.4em\relax IEEE, 2018, pp.
  149--154.

\bibitem{noureddine2016game}
M.~A. Noureddine, A.~Fawaz, W.~H. Sanders, and T.~Ba{\c{s}}ar, ``A
  game-theoretic approach to respond to attacker lateral movement,'' in
  \emph{International Conference on Decision and Game Theory for
  Security}.\hskip 1em plus 0.5em minus 0.4em\relax Springer, 2016, pp.
  294--313.

\bibitem{purvine2016graph}
E.~Purvine, J.~R. Johnson, and C.~Lo, ``A graph-based impact metric for
  mitigating lateral movement cyber attacks,'' in \emph{Proceedings of the 2016
  ACM Workshop on Automated Decision Making for Active Cyber Defense}, 2016,
  pp. 45--52.

\bibitem{HuangAPT}
L.~Huang and Q.~Zhu, ``A dynamic games approach to proactive defense strategies
  against advanced persistent threats in cyber-physical systems,''
  \emph{Computers \& Security}, vol.~89, p. 101660, 2020.

\bibitem{huang2018PER}
------, ``Adaptive strategic cyber defense for advanced persistent threats in
  critical infrastructure networks,'' \emph{ACM SIGMETRICS Performance
  Evaluation Review}, 2018.

\bibitem{alsaleh2018verifying}
M.~N. Alsaleh, E.~Al-Shaer, and Q.~Duan, ``Verifying the enforcement and
  effectiveness of network lateral movement resistance techniques.'' 2018.

\bibitem{shi2019quantitative}
Y.~Shi, X.~Chang, R.~J. Rodr{\'\i}guez, Z.~Zhang, and K.~S. Trivedi,
  ``Quantitative security analysis of a dynamic network system under lateral
  movement-based attacks,'' \emph{Reliability Engineering \& System Safety},
  vol. 183, pp. 213--225, 2019.

\bibitem{pawlick2019optimal}
J.~Pawlick, T.~T.~H. Nguyen, E.~Colbert, and Q.~Zhu, ``Optimal timing in
  dynamic and robust attacker engagement during advanced persistent threats,''
  in \emph{2019 International Symposium on Modeling and Optimization in Mobile,
  Ad Hoc, and Wireless Networks (WiOPT)}.\hskip 1em plus 0.5em minus
  0.4em\relax IEEE, 2019, pp. 1--8.

\bibitem{huang2019adaptive}
L.~Huang and Q.~Zhu, ``Adaptive honeypot engagement through reinforcement
  learning of semi-markov decision processes,'' in \emph{International
  Conference on Decision and Game Theory for Security}.\hskip 1em plus 0.5em
  minus 0.4em\relax Springer, 2019, pp. 196--216.

\bibitem{huang2018analysis}
------, ``Analysis and computation of adaptive defense strategies against
  advanced persistent threats for cyber-physical systems,'' in
  \emph{International Conference on Decision and Game Theory for
  Security}.\hskip 1em plus 0.5em minus 0.4em\relax Springer, 2018, pp.
  205--226.

\bibitem{huang2020game}
------, ``Game of duplicity: A proactive automated defense mechanism by
  deception design,'' \emph{arXiv preprint arXiv:2006.07942}, 2020.

\bibitem{goldberg2017cognitive}
I.~Goldberg, J.~R. Kozloski, C.~A. Pickover, N.~Sondhi, and M.~Vukovic,
  ``Cognitive honeypot,'' Jan.~31 2017, uS Patent 9,560,075.

\bibitem{horak2019optimizing}
K.~Hor{\'a}k, B.~Bo{\v{s}}ansk{\`y}, P.~Tom{\'a}{\v{s}}ek, C.~Kiekintveld, and
  C.~Kamhoua, ``Optimizing honeypot strategies against dynamic lateral movement
  using partially observable stochastic games,'' \emph{Computers \& Security},
  vol.~87, p. 101579, 2019.

\bibitem{wang2015efficient}
S.~Wang, W.~Lin, Y.~Yang, X.~Xiao, and S.~Zhou, ``Efficient route planning on
  public transportation networks: A labelling approach,'' in \emph{Proceedings
  of the 2015 ACM SIGMOD International Conference on Management of Data}, 2015,
  pp. 967--982.

\bibitem{jiang2020reinforcement}
C.~Jiang and X.~Zhu, ``Reinforcement learning based capacity management in
  multi-layer satellite networks,'' \emph{IEEE Transactions on Wireless
  Communications}, 2020.

\bibitem{xu2019cybersecurity}
S.~Xu, ``Cybersecurity dynamics: A foundation for the science of
  cybersecurity,'' in \emph{Proactive and Dynamic Network Defense}.\hskip 1em
  plus 0.5em minus 0.4em\relax Springer, 2019, pp. 1--31.

\bibitem{kaynar2016taxonomy}
K.~Kaynar, ``A taxonomy for attack graph generation and usage in network
  security,'' \emph{Journal of Information Security and Applications}, vol.~29,
  pp. 27--56, 2016.

\bibitem{chen2019enterprise}
P.-Y. Chen, S.~Choudhury, L.~Rodriguez, A.~Hero, and I.~Ray, ``Enterprise cyber
  resiliency against lateral movement: A graph theoretic approach,''
  \emph{arXiv preprint arXiv:1905.01002}, 2019.

\bibitem{hagberg-2014-credential}
A.~Hagberg, A.~Kent, N.~Lemons, and J.~Neil, ``Credential hopping in
  authentication graphs,'' in \emph{2014 International Conference on
  Signal-Image Technology Internet-Based Systems ({SITIS})}.\hskip 1em plus
  0.5em minus 0.4em\relax IEEE Computer Society, Nov. 2014.

\bibitem{mell2006common}
P.~Mell, K.~Scarfone, and S.~Romanosky, ``Common vulnerability scoring
  system,'' \emph{IEEE Security \& Privacy}, vol.~4, no.~6, pp. 85--89, 2006.

\bibitem{nawrocki2016survey}
M.~Nawrocki, M.~W{\"a}hlisch, T.~C. Schmidt, C.~Keil, and J.~Sch{\"o}nfelder,
  ``A survey on honeypot software and data analysis,'' \emph{arXiv preprint
  arXiv:1608.06249}, 2016.

\end{thebibliography}

\end{document}